\documentclass[12pt,reqno]{amsart}
\usepackage{amsmath}
\usepackage{latexsym}
\usepackage{amsfonts}
\usepackage{amssymb}
\usepackage{color}
\usepackage{bbm,dsfont}
\usepackage{graphicx}
\usepackage[normalem]{ulem}
\usepackage{caption}
\usepackage{enumerate}
\usepackage{tensor}

\newtheorem{proposition}{Proposition}

\newtheorem{lemma}{Lemma}

\theoremstyle{definition}

\newcommand{\de}{\,{\rm d}}

\newcommand{\R}{\mathbb{R}} 
\newcommand{\C}{\mathbb{C}} 
\newcommand{\rational}{\mathbb Q} 
\newcommand{\real}{\mathbb R} 
\newcommand{\complex}{\mathbb C}

\newcommand{\half}{\tfrac{1}{2}} 
\newcommand{\mo}[1]{\left| #1 \right|} 

\newcommand{\hi}{\mathcal{H}} 
\newcommand{\hh}{\mathcal{H}} 
\renewcommand{\aa}{\mathcal{A}} 
\newcommand{\bb}{\mathcal{B}} 

\newcommand{\ip}[2]{\left\langle\,#1\,|\,#2\,\right\rangle} 

\newcommand{\kb}[2]{|#1\rangle\langle#2|} 
\newcommand{\no}[1]{\left\|#1\right\|} 
\newcommand{\tr}[1]{\mathrm{tr}\left[#1\right]} 

\newcommand{\id}{\mathbbm{1}} 

\newcommand{\lam}{\lambda}

\newcommand{\rank}{{\rm rank}} 

\newcommand{\prem}{\mathcal{P}} 
 
\newcommand{\e}{{\rm e}}

\renewcommand{\Re}{{\rm Re}\,}
\newcommand{\h}[1]{\mathcal{#1}}
\newcommand{\be}{\begin{equation}}
\newcommand{\ee}{\end{equation}}

\newcommand{ \cP}{\mathcal{P}}

\begin{document}\setlength{\arraycolsep}{2pt}

\title[]{How many orthonormal bases are needed to distinguish all pure quantum states?}

\begin{abstract}
We collect some recent results that together provide an almost complete answer to the question stated in the title. 
For the dimension $d=2$ the answer is three.
For the dimensions $d=3$ and $d\geq 5$ the answer is four.
For the dimension $d=4$ the answer is either three or four.
Curiously, the exact number in $d=4$ seems to be an open problem.
\end{abstract}

\author[]{Claudio Carmeli}

\address{\textbf{Claudio Carmeli}; DIME, Universit\`a di Genova, Via Magliotto 2, I-17100 Savona, Italy}
 \email{claudio.carmeli@gmail.com}

\author[]{Teiko Heinosaari}
\address{\textbf{Teiko Heinosaari}; Turku Centre for Quantum Physics, Department of Physics and Astronomy, University of Turku, Finland}
\email{teiko.heinosaari@utu.fi}

\author[]{Jussi Schultz}
\address{\textbf{Jussi Schultz}; Dipartimento di Matematica, Politecnico di Milano, Piazza Leonardo da Vinci 32, I-20133 Milano, Italy, and Turku Centre for Quantum Physics, Department of Physics and Astronomy, University of Turku, Finland}
\email{jussi.schultz@gmail.com}

\author[]{Alessandro Toigo}
\address{\textbf{Alessandro Toigo}; Dipartimento di Matematica, Politecnico di Milano, Piazza Leonardo da Vinci 32, I-20133 Milano, Italy, and I.N.F.N., Sezione di Milano, Via Celoria 16, I-20133 Milano, Italy}
\email{alessandro.toigo@polimi.it}

\maketitle

\section{Introduction}

How many different measurement settings are needed in order to uniquely
determine a pure quantum state, and how should such measurements be chosen?
This problem goes back to a famous remark by W.~Pauli \cite{Pauli33}, in which he raised the question whether or not the position and the momentum distributions are enough to define the wave function uniquely modulo a global phase.
The original \emph{Pauli problem} has a negative answer \cite{PFQM44}, but it has evolved into many interesting variants and has been studied from several fruitful perspectives.
Discussion of the vast literature lies outside the scope of this work. 
Instead, we will concentrate only on a specific form of the Pauli problem, which is concerned with the minimal number of orthonormal bases, or projective measurements, in a finite dimensional Hilbert space that is needed in order to distinguish all pure quantum states. 
We require that {\em all} pure states are determined, so schemes that allow the determination of merely almost all pure states are outside of the scope of this work, even if they are interesting and important from the practical point of view. 
The purpose of this paper is to present the essential results related to our question in a comprehensible way.

It is quite obvious that a single orthonormal basis cannot distinguish all pure states in a $d$-dimensional Hilbert space, while it is known that with $d+1$ bases it is possible to distinguish all states, pure or mixed.
The problem of finding the minimal number of orthonormal bases determining an unknown pure state has been raised several times in the past.
It is easy to verify that the minimal number is three in dimension $2$, but in higher dimensions the problem becomes more difficult.
In 1978 A.~Vogt reported on R.~Wright's conjecture that three orthonormal bases are sufficient to identify an unknown pure state in any finite dimension \cite{Vogt78}.
In 1983 B.Z.~Moroz made the same claim and presented a proof for it \cite{Moroz83}, but in the erratum he explains that his proof does not work for all pure states and credits M.~Gromov for pointing out an argument that shows that at least four bases are needed in large dimensions \cite{Moroz84}.
In 1994 this argument was spelled out in greater detail by Moroz and A.M.~Perellomov \cite{MoPe94}, and they concluded that at least four bases are needed for any dimension $d\geq 9$. 
We will see that this conclusion can be extended to all dimensions $d\geq 5$ as a direct implication of the result of \cite{HeMaWo13}. 
The sufficiency of four generic orthonormal bases for unique pure state determination was found by D.~Mondragon and V.~Voroninski \cite{MoVo13}, and a concrete method for constructing four bases with the desired property  was recently introduced by P.~Jaming \cite{Jaming14}.

This paper is organized as follows.
We start by giving a precise mathematical formulation of the question in Sec.~\ref{sec:question}.
In Sec.~\ref{sec:qubit} we give an elementary argument for the fact that two orthonormal bases cannot distinguish all pure states in any dimension $d$.
Then in Sec.~\ref{sec:qutrit} we review the rank criterion first found in \cite{HeMaWo13} 
and explain how this implies that three bases are not enough in dimension $3$.
In Sec.~\ref{sec:4bases} we give a thorough presentation of Jaming's construction of four bases.
We continue in Sec.~\ref{sec:lower_bound} with the most technical part of this paper, which reviews the argument presented in \cite{HeMaWo13} that implies the impossibility of three bases in dimensions $d\geq 5$.

All these results together mean that the minimal number of orthonormal bases that are able to distinguish all pure quantum states is:
\begin{itemize}
\item three for $d=2$ 
\item four for $d=3$ and all dimensions $d\geq 5$
\item either three or four in $d=4$. 
\end{itemize}

Curiously, the final answer in $d=4$ still remains open.
In Sec.~\ref{sec:localbas} we rule out specific types of triples of bases, namely, those consisting solely of product vectors with respect to a splitting of the Hilbert space into a tensor product of $2$-dimensional spaces. 
Finally, in Sec.~\ref{sec:spin1} we treat spin-$1$ measurements to highlight the fact that even if four bases can distinguish all pure quantum states, these bases must be chosen appropriately and it may happen that some natural choices are not the best ones. We end this paper with a brief discussion of the problem in an infinite dimensional Hilbert space in Sec.~\ref{sec:discussion}.

\section{Formulation of the question}\label{sec:question}

Let $\hi$ be a finite $d$-dimensional Hilbert space.
The quantum states are described by density matrices, i.e., positive operators $\varrho$ on $\hi$ that satisfy $\tr{\varrho}=1$.
A quantum state $\varrho$ is \emph{pure} if it cannot be written as a mixture $\varrho = \half \varrho_1 + \half \varrho_2$ of two different states $\varrho_1$ and $\varrho_2$. 
Pure quantum states correspond to $1$-dimensional projections.
They can be alternatively described as rays of vectors in $\hi$, meaning that two vectors $\psi_1$ and $\psi_2$ correspond to the same pure state if there is a nonzero complex number $c$ such that $\psi_1=c\psi_2$.
For a unit vector $\psi\in\hi$, the corresponding density matrix is $\varrho=\kb{\psi}{\psi}$.

Let $\{ \varphi_j \}_{j=1}^d$ be an orthonormal basis of $\hi$.
(From now on, the term \emph{basis} means an orthonormal basis.)
If we perform a measurement of this basis in a state $\varrho$, then we get an outcome $j$ with the probability $\ip{\varphi_j}{\varrho \varphi_j}$.
The probability distribution $p(j)=\ip{\varphi_j}{\varrho \varphi_j}$ encodes the information that this measurement gives us about the unknown state $\varrho$. 
It is quite clear that this information is not enough to determine the input state uniquely.
For instance, the pure states corresponding to the unit vectors $1/\sqrt{2} (\varphi_1 \pm \varphi_2)$ are different but they lead to the same probability distribution.
If our aim is to identify an unknown quantum state uniquely, we should thus measure more than one orthonormal basis.

Let $\mathcal{B}_1=\{ \varphi^1_j \}_{j=1}^d, \ldots, \mathcal{B}_m=\{ \varphi^m_j \}_{j=1}^d$ be $m$  orthonormal bases of $\hi$.
We say that the bases $\mathcal{B}_1,\ldots,\mathcal{B}_m$ distinguish two different states $\varrho_1$ and $\varrho_2$ if 
\begin{align}\label{eq:can}
\ip{\varphi^\ell_j}{\varrho_1\varphi^\ell_j} \neq \ip{\varphi^\ell_j}{\varrho_2\varphi^\ell_j}
\end{align}
for some $\ell=1,\ldots,m$ and $j=1,\ldots,d$.
This means that if we get the probability distributions related to all $m$ bases in the states $\varrho_1$ and $\varrho_2$, then the measurement data corresponding to these two states are different. 

We say that the orthonormal bases $\mathcal{B}_1,\ldots,\mathcal{B}_m$ \emph{distinguish all pure states} if they distinguish any pair of different pure states.
Following \cite{BuLa89}, we also say that in this case the bases $\mathcal{B}_1,\ldots,\mathcal{B}_m$ are \emph{informationally complete with respect to pure states}.

We recall that it is possible to find finite collections of orthonormal bases that can distinguish all pure states.
Namely, it is known that there exist $d+1$ orthonormal bases which distinguish \emph{all} states, pure or mixed; see e.g. \cite{KoPa13}  for a construction.
Less than $d+1$ bases cannot distinguish all states simply because not enough parameters are determined. 
However, the pure states form a non-convex subset of all states, and one cannot thus rule out the possibility that less than $d+1$ bases can distinguish all pure states.
This leads to our main question stated in the title:

\emph{How many orthonormal bases are needed in order to distinguish all pure states?}

One may be tempted to criticize this question on the grounds that in any real experiment the states are perhaps never completely pure. 
But as any problem of this type, also this should be considered as a question on the fundamental limits of quantum theory.
As such, we believe that it reveals an interesting and even surprising aspect of the duality of states and measurements.

\section{Qubit and insufficiency of two bases}\label{sec:qubit}

As a warm up, let us consider the case of a qubit, i.e., $d=2$.
It is well known and explained also in almost any textbook that one can choose three orthonormal bases such that the related measurement outcome distributions determine an unknown qubit state uniquely.
This is rather obvious if one looks at the Bloch representation $\varrho_{\vec{r}}= \frac{1}{2}(\id + \vec{r}\cdot\vec{\sigma})$ of qubit states, where $\vec{r}$ is a vector in $\real^3$ satisfying $\no{\vec{r}}\leq 1$ and $\vec{\sigma}=(\sigma_x,\sigma_y,\sigma_z)$ consists of Pauli matrices.
Since 
\begin{equation*}
\vec{r}= (\tr{\varrho \sigma_x},\tr{\varrho \sigma_y},\tr{\varrho \sigma_z}) \, , 
\end{equation*}
we conclude that measurements of the eigenbases of $\sigma_x$, $\sigma_y$ and $\sigma_z$ specify the vector $\vec{r}$ and hence also the state $\varrho_{\vec{r}}$. 
More generally, if we fix three linearly independent unit vectors $\vec{a}$, $\vec{b}$ and $\vec{c}$ in $\real^3$, then the measurements of the eigenbases of $\vec{a}\cdot\vec{\sigma}$, $\vec{b}\cdot\vec{\sigma}$ and $\vec{c}\cdot\vec{\sigma}$ distinguish all states. 

A qubit state $\varrho_{\vec{r}}$ is pure exactly when $\no{\vec{r}}=1$. 
The direction of a unit vector $\vec{r}$ depends on two parameters only, so one may wonder if two orthonormal bases can suffice to determine any pure qubit state.
This is not true, in fact, in any dimension:

\begin{proposition}\label{prop:2bases}
In any dimension $d\geq 2$, two orthonormal bases cannot distinguish all pure states.
\end{proposition}

\begin{proof}
Our proof of this statement has been motivated by Theorem 1 in \cite{Finkelstein04}.
Let $\mathcal{B}_1=\{ \varphi_j \}_{j=1}^d$ and $\mathcal{B}_2=\{ \phi_k \}_{k=1}^d$ be two orthonormal bases of a $d$-dimensional Hilbert space. 
We need to find two nonparallel unit vectors $\psi_+$ and $\psi_-$ such that
\begin{equation}\label{eq:same-for-pm}
\mo{\ip{\xi}{\psi_+}}^2 = \mo{\ip{\xi}{\psi_-}}^2
\end{equation}
for all vectors $\xi \in \mathcal{B}_1 \cup \mathcal{B}_2$.

Let $\eta\in\hi$ be a unit vector orthogonal to $\varphi_1$.
We set $\psi_\pm=\frac{1}{\sqrt{2}}(\varphi_1 \pm \eta)$.
Then
\begin{align*}
\mo{\ip{\psi_\pm}{\xi}}^2 = \half \left( \mo{\ip{\varphi_1}{\xi}}^2 + \mo{\ip{\eta}{\xi}}^2 \right)
\pm \Re \left(\ip{\varphi_1}{\xi}\ip{\xi}{\eta}\right)\, ,
\end{align*}
so that
\begin{align*}
\mo{\ip{\psi_+}{\xi}}^2 - \mo{\ip{\psi_-}{\xi}}^2 = 2\ \Re \left(\ip{\varphi_1}{\xi}\ip{\xi}{\eta}\right) \, .
\end{align*}
Since $\ip{\varphi_1}{\varphi_j}\ip{\varphi_j}{\eta}=0$ for all $\varphi_j\in\mathcal{B}_1$, the condition \eqref{eq:same-for-pm} holds for all $\xi \in \mathcal{B}_1 \cup \mathcal{B}_2$ if 
\begin{equation}\label{eq:re1}
\Re \left(\ip{\varphi_1}{\phi_k}\ip{\phi_k}{\eta}\right)=0
\end{equation}
for all $\phi_k\in\mathcal{B}_2$.
The remaining thing is to show that it is possible to choose a unit vector $\eta\in\hi$ such that $\eta$ is orthogonal to $\varphi_1$ and \eqref{eq:re1} holds for all $\phi_k\in\mathcal{B}_2$.

Firstly, suppose that $\ip{\varphi_1}{\phi_1}=0$ or $\ip{\varphi_1}{\phi_2}=0$.  Then the corresponding choice $\eta=\phi_1$ or $\eta=\phi_2$ implies that  \eqref{eq:re1} holds for all $\phi_k\in\mathcal{B}_2$. 
If otherwise $\ip{\varphi_1}{\phi_1}\neq 0$ and $\ip{\varphi_1}{\phi_2}\neq 0$, we then set
\begin{equation*}
\eta = \frac{i\,|\ip{\varphi_1}{\phi_1}\ip{\varphi_1}{\phi_2}|}{\sqrt{|\ip{\varphi_1}{\phi_1}|^2 + |\ip{\varphi_1}{\phi_2}|^2}} (\ip{\varphi_1}{\phi_1}^{-1} \phi_1 - \ip{\varphi_1}{\phi_2}^{-1} \phi_2) \, .
\end{equation*}
It is easy to verify that $\eta$ is orthogonal to $\varphi_1$.
Furthermore, we get
\begin{align*}
& \ip{\varphi_1}{\phi_1}\ip{\phi_1}{\eta} = - \ip{\varphi_1}{\phi_2}\ip{\phi_2}{\eta} = \frac{i\,|\ip{\varphi_1}{\phi_1}\ip{\varphi_1}{\phi_2}|}{\sqrt{|\ip{\varphi_1}{\phi_1}|^2 + |\ip{\varphi_1}{\phi_2}|^2}} \\
& \ip{\varphi_1}{\phi_k}\ip{\phi_k}{\eta} = 0 \quad \textrm{for $k\geq 3$} \, , 
\end{align*}
hence \eqref{eq:re1} holds for all $\phi_k\in\mathcal{B}_2$.
\end{proof}

\section{Rank criterion and qutrit}\label{sec:qutrit}

Let $\mathcal{B}_1=\{ \varphi^1_j \}_{j=1}^d, \ldots, \mathcal{B}_m=\{ \varphi^m_j \}_{j=1}^d$ be $m$ orthonormal bases of $\hi$.
For each vector $\varphi^\ell_j$, we denote $P^\ell_j = \kb{\varphi^\ell_j}{\varphi^\ell_j}$.
Each $P^\ell_j$ is thus a $1$-dimensional projection.
Using this notation we observe that the orthonormal bases $\mathcal{B}_1,\ldots,\mathcal{B}_m$ \emph{cannot} distinguish two different states $\varrho_1$ and $\varrho_2$ if and only if
\begin{align}\label{eq:cannot}
\tr{P^\ell_j (\varrho_1 - \varrho_2)} = 0 \quad \text{for all $\ell=1,\ldots,m$ and $j=1,\ldots,d$} \,.
\end{align}
This condition means that the operator $\varrho_1 - \varrho_2$ is orthogonal to all the projections $P^\ell_j$ in the Hilbert-Schmidt inner product. (We recall that the Hilbert-Schmidt inner product of two operators $A$ and $B$ is 
\(
\ip{A}{B}_{HS} = \tr{A^\ast B}
\)).
Note that the operator $\varrho_1 - \varrho_2$ is selfadjoint and traceless.
Moreover, if $\varrho_1$ and $\varrho_2$ are pure states, then $\varrho_1 - \varrho_2$ has rank $2$.

The previous observation can be developed into a useful criterion when we look at all selfadjoint operators that are orthogonal to the projections $P^\ell_j$.
Suppose $T$ is a nonzero selfadjoint operator satisfying
\begin{equation*}
\tr{P^\ell_j T} = 0 \quad \text{for all $\ell=1,\ldots,m$ and $j=1,\ldots,d$} \,.
\end{equation*}
First of all, as $\sum_j P^\ell_j = \id$, the operator $T$ satisfies $\tr{T}=0$.
To derive some other properties of $T$, we write $T$ in the spectral decomposition
\begin{equation*}
T = \sum_{i=1}^{p} \lambda^+_i \kb{\psi^+_i}{\psi^+_i} - \sum_{i=1}^{n} \lambda^-_i \kb{\psi^-_i}{\psi^-_i} \, , 
\end{equation*}
where $\lambda^+_1,\ldots,\lambda^+_p$ and $-\lambda^-_1,\ldots,-\lambda^-_n$ are the strictly positive and strictly negative eigenvalues of $T$, respectively, and $\psi^+_1,\ldots,\psi^+_p,\psi^-_1,\ldots,\psi^-_n$ are orthogonal unit vectors. 
From $\tr{T}=0$ it follows that 
\begin{equation*}
\sum_{i=1}^{p} \lambda^+_i = \sum_{i=1}^{n} \lambda^-_i \equiv \lambda \, .
\end{equation*}
The rank of $T$ is $p+n$, the sum of its nonzero eigenvalues counted by their multiplicities.
As $T\neq 0$ and $\tr{T}=0$, $T$ must have both positive and negative eigenvalues, meaning that $n\geq 1$ and $p\geq 1$. Therefore, the rank of $T$ is at least $2$.
If the rank of $T$ is $2$, then $n=p=1$ and thus
\begin{equation*}
\lambda^{-1} T = \kb{\psi^+_1}{\psi^+_1} - \kb{\psi^-_1}{\psi^-_1} \, .
\end{equation*}
This implies that the orthonormal bases $\mathcal{B}_1,\ldots,\mathcal{B}_m$ cannot distinguish the pure states $\varrho_1=\kb{\psi^+_1}{\psi^+_1}$ and $\varrho_2=\kb{\psi^-_2}{\psi^-_2}$

Our previous discussion can be summarized in the form of the following criterion.

\begin{proposition}\label{prop:rank} 
Orthonormal bases $\mathcal{B}_1,\ldots,\mathcal{B}_m$ can distinguish all pure states if and only if every nonzero selfadjoint operator $T$ that satisfies
\begin{equation}\label{eq:cannotT}
\tr{P^\ell_j T} = 0 \quad \text{for all $\ell=1,\ldots,m$ and $j=1,\ldots,d$}
\end{equation}
has rank at least $3$.
\end{proposition}

Using this criterion we can prove the following statement.

\begin{proposition}\label{prop:critdim3}
In dimension $d=3$,
\begin{enumerate}[(i)]
\item three orthonormal bases cannot distinguish all pure states;
\item four orthonormal bases can distinguish all pure states if and only if the Hilbert-Schmidt orthogonal complement of the projections $\{P^\ell_j\mid\ell=1,2,3,4,\,j=1,2,3\}$ is either $\{0\}$ or is the linear span of a single nonzero and invertible selfadjoint operator.
\end{enumerate}
\end{proposition}
\begin{proof}
Our proof is adapted from the analogous one of \cite[Proposition 5]{HeMaWo13}. 
By Proposition \ref{prop:rank}, $m$ orthonormal bases $\mathcal{B}_1,\ldots,\mathcal{B}_m$ distinguish all pure states if and only if every nonzero selfadjoint operator $T\in\{P^\ell_j\mid\ell=1,\ldots,m,\,j=1,2,3\}^\perp$ is invertible. We claim that in this case there cannot exist two linearly independent selfadjoint operators $T_1,T_2\in\{P^\ell_j\mid\ell=1,\ldots,m,\,j=1,2,3\}^\perp$.\\
Indeed, suppose on the contrary that $T_1$ and $T_2$ are two such operators. Since $\det(T_1)$ and $\det(T_2)$ are nonzero, there are real numbers $\alpha_1,\alpha_2$ such that $\det{(\alpha_1 T_1)}>0$ and $\det{(\alpha_2 T_2)} <0$. By linear independence, the convex combination $\lambda \alpha_1 T_1 + (1-\lambda) \alpha_2 T_2$ is nonzero for all $\lambda\in[0,1]$. The determinant is a continuous function, and hence the intermediate value theorem implies that $\det{[\lambda_0 \alpha_1 T_1 + (1-\lambda_0) \alpha_2 T_2]} =0$ for some $\lambda_0 \in (0,1)$.
Therefore, the nonzero selfadjoint operator $T=\lambda_0 \alpha_1 T_1 + (1-\lambda_0) \alpha_2 T_2$ is not invertible. But $T$ satisfies \eqref{eq:cannotT}, which then contradicts Proposition \ref{prop:rank}.\\
We now come to the proof of (i). Three orthonormal bases $\mathcal{B}_1,\mathcal{B}_2,\mathcal{B}_3$ give $9$ projections $P^\ell_j$.
However, as $\sum_{j=1}^3 P^\ell_j = \id$ for each $\ell=1,2,3$, at most $7$ of them are linearly independent. The dimension of the real vector space of all selfadjoint operators is $9$, hence we conclude that there are at least $2$ linearly independent selfadjoint operators that are orthogonal to all projections $P^\ell_j$. The previous claim then implies that $\mathcal{B}_1,\mathcal{B}_2,\mathcal{B}_3$ cannot distinguish all pure states.\\
To prove item (ii), observe that, if the linear space $\{P^\ell_j\mid\ell=1,2,3,4,\,j=1,2,3\}^\perp$ has dimension $k$, then we can find a basis of it consisting of selfadjoint operators. Indeed, if $T_1,\ldots,T_k$ is any linear basis, then the selfadjoint operators $T^+_1,\ldots,T^+_k,T^-_1,\ldots,T^-_k$ given by
$$
T^+_h = T+T^*,\qquad T^-_h = i(T-T^*)
$$
still satisfy \eqref{eq:cannotT} and generate the linear space $\{P^\ell_j\mid\ell=1,2,3,4,\,j=1,2,3\}^\perp$. Extracting $k$ linearly independed elements from these operators, we get a basis of selfadjoint operators.  Therefore, by our earlier claim the four bases $\mathcal{B}_1,\dots,\mathcal{B}_4$ can distinguish all pure states only if either $k=0$ or $k=1$. In the latter case, $\{P^\ell_j\mid\ell=1,2,3,4,\,j=1,2,3\}^\perp = \C T$ for some selfadjoint operator $T$, which must then be invertible by Proposition \ref{prop:rank}. Conversely, the sufficiency of these two conditions is clear by Proposition \ref{prop:rank}.
\end{proof}

In the next section, we will see that in dimension $d=3$ actually there exist four orthonormal bases distinguishing all pure states. The condition in item (ii) of Proposition \ref{prop:critdim3} is then very useful to explicitely construct such bases. As an example, Section \ref{sec:spin1} below will provide an application to the measurement of the orthonormal bases corresponding to four different spin directions in a spin-$1$ system.

\section{Four bases that distinguish all pure states}\label{sec:4bases}
Up to now we have seen that already in dimension 3, we can never find three orthonormal bases which would yield unique determination of all pure states. Therefore, the minimal number of bases in that case is at least four. In this section we show that, perhaps surprisingly, four properly chosen orthonormal bases are sufficient regardless of the dimension of the Hilbert space.

\begin{proposition}\label{prop:4bases} 
In any finite dimension, there exist four orthonormal bases $\mathcal{B}_1,\mathcal{B}_2,\mathcal{B}_3,\mathcal{B}_4$ that can distinguish all pure states. 
\end{proposition}

The proof is based on an explicit construction of the bases in the Hilbert space $\hh=\C^d$, as presented by Jaming in \cite{Jaming14}. 
His construction uses properties of the Hermite polynomials, but a similar construction works also for any other sequence of orthogonal polynomials. Different polynomials will lead to different bases, so this freedom in choosing the polynomials may be sometimes useful.

The construction begins by fixing a sequence of orthogonal polynomials.
By a \emph{sequence of orthogonal polynomials} we mean a sequence $p_0,p_1,p_2,\ldots$ of real polynomials such that the degree of $p_n$ is $n$ and
\begin{align*}
\int_a^b p_n(x)p_\ell(x) \ w(x)\, \de x =  \delta_{n\ell} 
\end{align*}
for a nonnegative weight function $w$ and either finite or infinite interval $[a,b]$.
The most common sequences of orthogonal polynomials are the (normalized versions of) Chebyshev, Hermite,  Laguerre and Legendre polynomials.
For instance, the $n^{\rm th}$ Hermite polynomial $H_n$ is defined by the formula 
\begin{equation*}
H_n(x)  =  \frac{(-1)^n}{\sqrt{2^n n!}}  \e^{x^2} \frac{\de^n}{\de x^n} \e^{-x^2}.
\end{equation*}
The Hermite polynomials form a sequence of orthogonal polynomials with respect to the weight function $w(x)=\frac{1}{\sqrt{\pi}} e^{-x^2}$ and the infinite interval $(-\infty,\infty)$.

The following construction uses three basic properties shared by any sequence of orthogonal polynomials.
Let $c_n$ denote the highest coefficient of a polynomial $p_n$.
It can be shown (see e.g. \cite[pp. 43-46]{Szego}) that the following properties hold:
\begin{enumerate}[(a)]
\item The roots of $p_n$ are all real and distinct.
\item $p_n$ and $p_{n+1}$ have no common roots.
\item For all $x\neq y$, the \emph{Christoffel-Darboux formula} holds:
\begin{equation*}
\sum_{j=0}^{n} p_j(x) p_j(y) = \frac{c_n}{c_{n+1}} \frac{p_{n+1}(x)p_{n}(y) - p_{n}(x)p_{n+1}(y)}{x-y} \, .
\end{equation*}
\end{enumerate}
In dimension $d$, only the first $d+1$ polynomials of the sequence will be needed.

To construct the first basis, let $x_1,\ldots, x_d$ be the roots of the polynomial $p_d$, and define 
\begin{equation}\label{eqn:bases_1}
\widetilde{\varphi}^1_j= (p_0(x_j),p_1(x_j),\ldots, p_{d-1}(x_j))^T
\end{equation}
for $j=1,\ldots, d$.
Then each vector $\widetilde{\varphi}_j^1$ is nonzero since $p_0(x_j)\neq 0$ ($p_0$ is a nonzero constant polynomial), and since the zeros are all distinct, we may apply the Christoffel-Darboux formula to get
\begin{eqnarray*}
\langle \widetilde{\varphi}_i^1\vert \widetilde{\varphi}_j^1\rangle &=& \sum_{k=0}^{d-1} p_k(x_i) p_k(x_j) \\
&=& \frac{c_{d-1}}{c_d} \frac{p_d(x_i)p_{d-1}(x_j) - p_{d-1}(x_i)p_d(x_j)}{x_i-x_j} =0
\end{eqnarray*}
for $i\neq j$. 
Thus, the vectors are orthogonal. 
By normalizing $\varphi_j^1 = \Vert \widetilde{\varphi}_j^1\Vert^{-1}\widetilde{\varphi}_j^1$ we obtain an orthonormal basis $\mathcal{B}_1=\{\varphi_j^1 \}_{j=1}^d$ of $\C^d$.

Let then $y_1,\ldots, y_{d-1}$ be the roots of the polynomial $p_{d-1}$ and define
\begin{equation}\label{eqn:bases_2}
\widetilde{\varphi}_j^2= (p_0(y_j),p_1(y_j),\ldots, p_{d-1}(y_j))^T.
\end{equation}
for $j=1,\ldots, d-1$. These vectors are again nonzero and orthogonal, and since the $y_j$'s are the roots of $p_{d-1}$, the last component is $ p_{d-1}(y_j)=0$. 
Hence, we can again normalize $\varphi_j^2 = \Vert \widetilde{\varphi}_j^2\Vert^{-1}\widetilde{\varphi}_j^2$  and define $\varphi_d^2=(0,\ldots, 0,1)^T$ to obtain another orthonormal basis $\mathcal{B}_2 = \{ \varphi_j^2 \}_{j=1}^d$.

For the two remaining bases, we first pick a number $\alpha\in\R$ which is not a rational multiple of $\pi$.
We then define 
\begin{equation}\label{eqn:bases_3}
\widetilde{\varphi}_j^3= (p_0(x_j), \e^{i\alpha}p_1(x_j),\ldots, \e^{i(d-1)\alpha}p_{d-1}(x_j))^T,
\end{equation}
which, after normalization, gives the third basis $\mathcal{B}_3 = \{   \varphi_j^3 \}_{j=1}^d$. 
Finally, we set 
\begin{equation}\label{eqn:bases_4}
\widetilde{\varphi}_j^4= (p_0(y_j),\e^{i\alpha}p_1(y_j),\ldots, \e^{i(d-1)\alpha}p_{d-1}(y_j))^T
\end{equation}
which, after normalizing and adding the vector $\varphi_d^4=(0,\ldots, 0,1)^T$, gives the last basis $\mathcal{B}_4 = \{   \varphi_j^4 \}_{j=1}^d$.

Using the basis $\bb_1,\ldots,\bb_4$, we can now prove our main result.

\begin{proof}[Proof of Proposition \ref{prop:4bases}]
As usual, we denote $P_j^\ell = \vert \varphi^\ell_j\rangle\langle \varphi^\ell_j\vert$ for $\ell=1,\ldots,4$ and $j=1,\ldots,d$.
By \eqref{eq:cannot}, in order to prove that the bases $\bb_1,\ldots,\bb_4$ can distinguish all pure states, we need to show that, for any two pure states $\varrho_1 = \kb{\xi}{\xi}$ and $\varrho_2 = \kb{\eta}{\eta}$, the condition
\begin{equation}\label{eq:ortoT}
\tr{P^\ell_j (\varrho_1 - \varrho_2)} = 0 \quad \text{for all $\ell=1,\ldots,m$ and $j=1,\ldots,d$}
\end{equation}
implies $\xi = \e^{i\theta}\eta$ for some $\theta\in\R$.
To see this, let $T = \varrho_1 - \varrho_2$, and write $T=(t_{ij})_{i,j=1}^{d}$ with respect to the standard basis of $\C^d$. The selfadjointness of $T$ implies that $t_{ii}\in\R$ and $t_{ji}=\overline{t_{ij}}$.
Using the explicit form of the vectors from \eqref{eqn:bases_1}--\eqref{eqn:bases_4}, the orhogonality condition \eqref{eq:ortoT} then yields
\begin{align}
& \sum_{k,l=0}^{d-1} t_{k+1,l+1} \ p_k(z)p_l(z) = 0 \label{eqn:polynomial-1} \\
& \sum_{k,l=0}^{d-1} t_{k+1,l+1} \e^{i(l-k)\alpha}\ p_k(z)p_l(z) = 0\label{eqn:polynomial-2}
\end{align}
for every $z\in\{x_1,\ldots, x_d, y_1,\ldots, y_{d-1}\}$. 
The degree of $p_n$ is $n$, so the polynomials in \eqref{eqn:polynomial-1} and \eqref{eqn:polynomial-2} have degree at most $2d-2$.
But the above equations state that these polynomials have $2d-1$ distinct roots, so they are actually identically zero.
Therefore,  \eqref{eqn:polynomial-1} and \eqref{eqn:polynomial-2} hold for all $z\in\real$.\\
We can now look at the highest order terms in \eqref{eqn:polynomial-1} and \eqref{eqn:polynomial-2}. 
This corresponds to  $k=l=d-1$ so that by linear independence of the monomials $1,z,z^2,\ldots,z^{2d-2}$ we have  $t_{d,d}=0$.
Since
$$
t_{ij} = \xi_i\overline{\xi_j}  -  \eta_i \overline{\eta_j} \quad \text{for all $i,j$} \,,
$$
it follows that $\vert \xi_{d}\vert^2 = \vert \eta_{d}\vert^2$ so that $\xi_{d} =\e^{i\theta}\eta_{d}$ for some $\theta\in\R$. 
Assume for the moment that $\xi_d\neq 0$.\\
We can now consider the terms of order $2d-3$. Since they appear only for $k=d-1$ and $l=d-2$ or vice versa, we must have
\begin{equation*}
t_{d,d-1}+t_{d-1,d}=t_{d,d-1} \e^{-i\alpha} + t_{d-1,d}\e^{i\alpha}=0 \, .
\end{equation*} 
By substituting $t_{d-1,d} = \overline{t_{d,d-1}}$ we then get 
\begin{equation*}
{\rm Re}\, (t_{d,d-1} )= {\rm Re }\,(t_{d,d-1} \e^{-i\alpha} )=0 \, .
\end{equation*}
Since $\e^{-i\alpha}\notin\R$, we have $t_{d,d-1}=0$ which implies that $\xi_{d}\overline{\xi_{d-1}}  =  \eta_{d} \overline{\eta_{d-1}}$. By our assumption, $\xi_{d-1}=\e^{i\theta} \eta_{d-1}$.\\
We now proceed by induction. Suppose that $\xi_{d-p}=\e^{i\theta}\eta_{d-p}$ for $p=0,\ldots, n-1$. It follows that $t_{d-p,d-q}=0$ for $p,q\leq n-1$, so that the highest order terms in \eqref{eqn:polynomial-1} and \eqref{eqn:polynomial-2} are of order $2d-n-2$, and they appear only with $k=d-1$ and $l=d-n-1$ or vice versa. This gives us 
\begin{align*}
{\rm Re}\, (t_{d,d-n})={\rm Re}\, (t_{d,d-n}\e^{-i\alpha})=0
\end{align*}
so that $t_{d,d-n}=0$. Hence, $\xi_{d-n}=\e^{i\theta}\eta_{d-n}$.\\
Finally, if we have $\xi_{d}=0$ so that also $\eta_{d}=0$, then $t_{d,n}=0$ and $t_{n,d}=0$ for all $n$, and the summations in \eqref{eqn:polynomial-1} and \eqref{eqn:polynomial-2} terminate at $d-2$. Hence, we can repeat the above process starting from the highest order term which is now of order $2d-4$. By induction, if $\xi_{d}=\ldots=\xi_{d-(n-1)}=0$ but $\xi_{d-n}\neq 0$, then also $\eta_{d}=\ldots=\eta_{d-(n-1)}=0$ so that the process can be started from the terms of order  $2(d-n)-2$ which would give $\xi_{d-n}=\e^{i\theta}\eta_{d-n}$ and then proceed as before. 
This completes the proof of Proposition \ref{prop:4bases}.
\end{proof}

\section{Insufficiency of three bases in dimension $5$ and higher }\label{sec:lower_bound}

We have seen that in every finite dimension $d=2,3,\ldots$, it is possible to choose a set of four orthonormal bases that distinguish all pure states, but no pair of orthonormal bases can have this property.
Can a set of three orthonormal bases distinguish all pure states?
As we know already, this question has a positive answer in $d=2$ and a negative answer in $d=3$.
In the following, we explain how topological considerations lead to a negative answer for all dimensions $d\geq 5$.
For more details on the topological aspects of state determination we refer to \cite{HeMaWo13} and \cite{KeVrWo15}.

First, it is useful to generalize the property of distinguishing pure states to arbitrary sets $\mathcal{A}=\{A_1,\ldots,A_n\}$ of selfadjoint operators on $\hi$. We say that such a set $\mathcal{A}$ {\em distinguishes all pure states} if for two different pure states $\varrho_1$ and $\varrho_2$, there exists \(A_i\in\mathcal{A}\) such that
\begin{align}\label{eq:can2}
\tr{A_i \varrho_1} \neq \tr{A_i \varrho_2} \, .
\end{align}
Clearly, if $\mathcal{A}$ consists  of the orthogonal projections defined by a collection of orthonormal bases, we obtain again \eqref{eq:can}. 

In what follows, we use the notation below: 
for each dimension $d=2,3,...$, we denote by
\begin{itemize}
\item $\mathfrak{s}_d$ the minimal number of \emph{selfadjoint operators} which distinguish all pure states;
\item $\mathfrak{b}_d$ the minimal number of \emph{orthonormal bases} which distinguish all pure states.
\end{itemize}

It is easy to see that $\mathfrak{s}_d$ gives a lower bound for $\mathfrak{b}_d$.
Namely, if we have $m$ bases, they give $m\cdot d$ projections. 
All the projections corresponding to a basis sum up to the identity operator $\id$; since $\tr{\id\varrho} = 1$ for all states $\varrho$, one projection for each basis can then be left out without losing any information.
We thus conclude that
\begin{equation}\label{eq:nd>md}
(d-1) \cdot \mathfrak{b}_d \geq \mathfrak{s}_d \, .
\end{equation}
Therefore, lower bounds for \( \mathfrak{s}_d\) translate  into lower bounds for $\mathfrak{b}_d$.

Let us denote by $\prem$ the set of pure states.
Saying that $\mathcal{A}$ distinguishes all pure states means that the map
\begin{equation}\label{eq::emb}
f^{\mathcal{A}}\colon \prem  \to \R^n, \qquad \varrho  \mapsto (\tr{A_1\varrho},\ldots,\tr{A_n\varrho})
\end{equation}
is injective. Hence, roughly speaking, the selfadjoint operators $\mathcal{A}$ distinguish all pure states if and only if the map $\eqref{eq::emb}$ identifies $\cP$ with a subset of the Euclidean space $\R^n$. This is a very useful observation: indeed, suppose, for instance, that it is possible to prove that there is a natural number $n_0$ such that no injective map $\cP\to \R^n$ exists if $n<n_0$; then one can conclude that $\mathfrak{s}_d\geq n_0$.

Up to this point our considerations were purely set theoretical. However, it can be proved that, if the map $f^\mathcal{A}$ is injective, then it is actually a smooth embedding (see Proposition \ref{prop:app2} of Appendix~A; we refer to \cite{Lee2009} for the necessary notions from differential geometry).
Therefore, the selfadjoint operators $\mathcal{A}$ can distinguish all pure states only if the manifold of pure states $\prem$ can be smoothly embedded in $\R^n$.

The minimal $n$ for which $\cP$ can be smoothly embedded in $\R^n$, is called the \emph{embedding dimension} of $\cP$.
From the heuristic point of view, if $M_1$ and $M_2$ are two smooth manifolds of the same dimension, we expect that the embedding dimension of $M_1$ is greater than the embedding dimension of $M_2$ if the shape of $M_1$ is more involved than the shape of $M_2$.
For instance, a $2$-dimensional torus can be smoothly embedded in $\R^3$ whereas the Klein bottle requires $\R^4$.

The problem of determining the embedding dimension of complex projective spaces has been studied extensively in the mathematical literature  and  the best bounds are, up to our knowledge,  those obtained in \cite{Mayer65}. 
They lead to the conclusion that
\begin{equation}\label{eq:lower-bound}
\mathfrak{s}_d \geq \left\{ \begin{array}{ll}
4d-2\alpha-3& \text{ for all }d\geq 2\\
4d-2\alpha-2 & \text{ for }d\text{ odd, and }\alpha=2\mod 4 \\
4d-2\alpha-1 & \text{ for }d\text{ odd, and }\alpha=3\mod 4
\end{array}\right.
\end{equation}
where $\alpha$ is the number of $1$'s in the binary expansion of $d-1$ (see \cite[Theorem 6]{HeMaWo13}). 

Writing the expansion $d-1 = \sum_{j=0}^n a_j 2^j$, we have that $2^n\leq d-1$ if and only if $n\leq \log_2(d-1)$. 
Hence, 
\begin{equation*}
\alpha\leq n+1 \leq \log_2(d-1)+1
\end{equation*}
so that 
\begin{equation*}
\mathfrak{s}_d \geq 4d-2\alpha-3 \geq 4d -2\log_2(d-1) - 5 
\end{equation*}
for all $d\geq 2$.
Inserting this into \eqref{eq:nd>md} we obtain
\begin{equation}\label{eqn:lower_bound}
\mathfrak{b}_d \geq \frac{4d-2\log_2(d-1)-5}{d-1} = 4-\frac{2\log_2(d-1)+1}{d-1} \, .
\end{equation} 
To see the consequences of the derived lower bound \eqref{eqn:lower_bound}, we examine the function 
\begin{equation*}
f:[2,\infty) \to \real \, , \quad f(x) =  4-\frac{2\log_2(x-1)+1}{x-1}
\end{equation*}
 (see Fig.~\ref{fig:minimum}). Firstly, we notice that $f(2) = 3$ and $\lim_{x\to\infty}f(x) = 4$. Secondly, the derivative $f'(x)$ has a single zero at $x_0=1+\frac{e}{\sqrt{2}}\simeq 2.922$, and it satisfies $f'(x) <0$ for $x\in(2,x_0)$ and $f'(x)>0$ for $x>x_0$. 
 Finally, $f(8)\simeq 3.055$ so that $f(x) > 3$ for all $x\geq 8$. 
Since $\mathfrak{b}_d$ is an integer and $\mathfrak{b}_d\geq f(d)$, we thus have $\mathfrak{b}_d \geq 4$ for all $d\geq 8$. 

\begin{figure}
\centering
            \includegraphics[width=10cm]{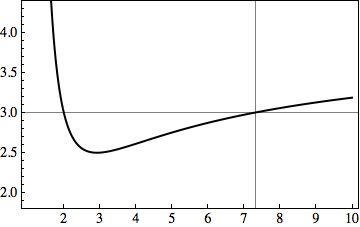} 
           \caption{The function $f(x)=4-\frac{2\log_2(x-1)+1}{x-1}$ gives a lower bound for the mimimal number of orthonormal bases.
           Since $\mathfrak{b}_d\geq f(d)$ for each integer $d\geq 2$, we conclude that $\mathfrak{b}_d \geq 4$ for all $d\geq 8$.}
    \label{fig:minimum}
\end{figure}

For the dimensions $d=2,\ldots, 7$ we need to analyse the lower bound more carefully. 
In the table below we have calculated the values of the lower bound for $\mathfrak{s}_d$ given in \eqref{eq:lower-bound}.
We thus observe that unless $d=2$ or $4$, three orthonormal bases cannot be sufficient. 

\begin{center}
\begin{tabular}{| c | c | c | c |}
\hline
$d$ & $\mathfrak{s}_d$ & $3(d-1)$\\
\hline\hline
2  & 3 & 3 \\
3  & 7 & 6 \\
4  & 9 & 9 \\
5  & 15 & 12 \\
6  & 17 & 15 \\
7 & 22 & 18 \\
\hline
\end{tabular}
\end{center}

In conclusion, for the dimensions $d=3$ and $d\geq 5$ the minimal number of bases is four.
For $d=2$ the minimal number is three and in the case $d=4$ it is, based on our knowledge, either three or four.

\section{Insufficiency of four product bases in dimension $d=4$}\label{sec:localbas}

In our search for the minimal number of orthonormal bases that can distinguish all pure states, the remaining question is:

\emph{Is it possible to find three orthonormal bases in dimension $4$ that can distinguish all pure states?}

Unfortunately, we are able to provide only a partial answer to this question.
Namely, in the following we will see that if we consider the splitting of the $4$-dimensional Hilbert space into a tensor product $\hi = \complex^2\otimes \complex^2$,  then for any three bases consisting solely of product vectors, the answer is negative. 
In fact, we will show that even four product bases are not enough.

Before we concentrate on dimension $4$, we slightly elaborate the state distinction criterion used in earlier sections.
As we have seen, a set $\mathcal{A}=\{A_1,\ldots,A_n\}$ of selfadjoint operators cannot distinguish two states $\varrho_1$ and $\varrho_2$ if and only if 
\begin{align*}
\tr{ A_j (\varrho_1 - \varrho_2)} = 0 \quad\text{for all $j=1,\ldots,n$}\,.
\end{align*}
Since the trace is a linear functional, this is equivalent to
\begin{align*}
\tr{ \Big(\sum_{j=1}^n \alpha_j A_j\Big) (\varrho_1 - \varrho_2)} = 0 \quad \text{for all $\alpha_1,\ldots,\alpha_n\in\C$}\,.
\end{align*}
Moreover, as $\tr{\varrho_1 - \varrho_2} = 0$, we can rewrite the previous condition as
\begin{align*}
\tr{ \Big(\alpha_0 \id + \sum_{j=1}^n \alpha_j A_j\Big) (\varrho_1 - \varrho_2)} = 0 \quad \text{for all $\alpha_0,\alpha_1,\ldots,\alpha_n\in\C$}\,.
\end{align*}
This equivalent formulation shows that, for the purpose of state distinction, we can always switch from $\mathcal{A}$ to the linear space spanned by $\mathcal{A}$ and $\id$.
We denote this subspace of operators by $\mathcal{R}(\mathcal{A})$, i.e., 
\begin{equation}
\label{eq:opsyst}
\mathcal{R}(\mathcal{A}) = \left\{\alpha_0 \id + \sum_{j=1}^n \alpha_j A_j \bigg\vert \alpha_j \in\C \right\}
\end{equation}
and call it the (complex) {\em operator system} generated by the selfadjoint operators $\mathcal{A}$. If the set $\aa$ consists of the projections corresponding to $m$ orthonormal bases $\bb_1 = \{\varphi^1_j\}_{j=1}^d,\ldots,\bb_m = \{\varphi^m_j\}_{j=1}^d$, we write also
$$
\h R(\bb_1,\ldots,\bb_m) = \h R(\{P^\ell_j\mid \ell = 1,\ldots,m,\,j=1,\ldots,d\}) \,,
$$
where as usual $P^\ell_j = \kb{\varphi^\ell_j}{\varphi^\ell_j}$. Our discussion then yields the following conclusion.
\begin{proposition}\label{prop:eqopsys}
Let $\mathcal{A}$ and $\mathcal{A'}$ be two sets of selfadjoint operators. If $\mathcal{R}(\mathcal{A}) = \mathcal{R}(\mathcal{A'})$, then $\aa$ and $\aa'$ distinguish the same pairs of states.
\end{proposition}

Let us make use of this fact to show that in the Hilbert space $\hh = \C^2\otimes \C^2$ four product bases cannot distinguish all pure states. By product basis, we mean an orthonormal basis of $\hh$ that is constructed from two orthonormal bases of $\C^2$ by taking their tensor product. 
More precisely, if $\{\varphi_1,\varphi_2\}$ and $\{\phi_1,\phi_2\}$ are two orthonormal bases of $\C^2$, then $\{\varphi_i\otimes \phi_j \mid i,j=1,2\}$ is an orthonormal basis of $\C^4$. 
From the physics point of view this corresponds to a scheme where two parties try to determine the unknown pure state of a composite system by performing only local measurements. 

Using the Bloch representation, any $1$-dimensional projection on $\complex^2$ can be written as   $P_{\vec{n}} = \frac{1}{2}(\id + \vec{n}\cdot \vec{\sigma})$, where $\vec{n}\in\R^3$ is a unit vector. 
Furthermore,  
\begin{equation*}
\tr{P_{\vec{m}} P_{\vec{n}} } = \frac{1}{2}(1+\vec{m}\cdot\vec{n}) \, .
\end{equation*}
Therefore, the $1$-dimensional projections corresponding to an orthonormal basis of $\C^2$ are $P_{\vec{n}}$ and $P_{-\vec{n}}$ where $\vec{n}$ is fixed by the choice of the basis. 

Suppose  that we have two quadruples of (not necessarily distinct) orthonormal bases $\mathcal{B}_1',\mathcal{B}_2', \mathcal{B}_3',\mathcal{B}_4'$ and $\mathcal{B}_1'',\mathcal{B}_2'', \mathcal{B}_3'',\mathcal{B}_4''$ of $\C^2$ with the corresponding quadruples of unit vectors $\vec{m}_1,\vec{m}_2,\vec{m}_3,\vec{m}_4$ and $\vec{n}_1, \vec{n}_2,\vec{n}_3,\vec{n}_4$. 
We define four product bases of $\C^4$ via 
\begin{equation*}
\mathcal{B}_j=\{\varphi\otimes\phi \mid \varphi\in\mathcal{B}_j', \phi\in\mathcal{B}_j'' \} \, .
\end{equation*}
The four projections 
\begin{align*}
P_{\vec{m}_j} \otimes P_{\vec{n}_j}\, , \quad P_{\vec{m}_j} \otimes P_{-\vec{n}_j}\, , \quad P_{-\vec{m}_j} \otimes P_{\vec{n}_j} \, , \quad P_{-\vec{m}_j} \otimes P_{-\vec{n}_j}
\end{align*}
corresponding to an orthonormal basis $\mathcal{B}_j$ have the same linear span as the selfadjoint operators 
\begin{align*}
\id\otimes\id \, , \quad \vec{m}_j\cdot\vec{\sigma} \otimes \id \, , \quad \id\otimes \vec{n}_j\cdot\vec{\sigma} \,  ,\quad \vec{m}_j\cdot\vec{\sigma} \otimes \vec{n}_j\cdot\vec{\sigma} \, .
\end{align*}
Therefore, by introducing the set
$$
\aa = \{\vec{m}_j\cdot\vec{\sigma} \otimes \id , \, \id\otimes \vec{n}_j\cdot\vec{\sigma} , \, \vec{m}_j\cdot\vec{\sigma} \otimes \vec{n}_j\cdot\vec{\sigma} \mid j=1,\ldots,4\} \,,
$$
we have the equality
$$
\h R (\bb_1,\ldots,\bb_4) = \h R(\aa) \, .
$$
We will next show that there exist two maximally entangled pure states $\varrho_1\neq\varrho_2$ which are not distinguished by the set $\aa$, thus implying by Proposition \ref{prop:eqopsys} that the bases $\bb_1,\ldots,\bb_4$ cannot distinguish all pure states.

Recall that a unit vector $\Omega\in\C^2\otimes \C^2$ is called \emph{maximally entangled} if $\Omega = \frac{1}{\sqrt{2}} (\varphi_1 \otimes \phi_1  + \varphi_2 \otimes \phi_2)$ for some orthonormal bases $\{\varphi_1,\varphi_2\}$ and $\{\phi_1,\phi_2\}$ of $\C^2$. If $\{e_1,e_2\}$ denotes the canonical basis of $\C^2$, and $\Omega_0= \frac{1}{\sqrt{2}}(e_1\otimes e_1 + e_2\otimes e_2)$, then there always exists a unitary operator $U$ on $\C^2$ such that $\Omega= (U\otimes \id) \Omega_0$ (see, e.g., \cite[Lemma 2]{Vollbrecht2000}). Since the global phase of $\Omega$ is irrelevant for our purposes, we  see  that the maximally entangled pure states are in one-to-one correspondence with elements of the quotient group $SU(2)/\{\pm I\}$, which in turn is diffeomorphic to $SO(3)$ through mapping $U\mapsto R_U$ given by the equality
\[
U^* \vec{x}\cdot\vec{\sigma}U = R_U\vec{x}\cdot\vec{\sigma} \quad \textrm{for all $\vec{x}\in\R^3$} \,.
\]

Suppose that we try to determine the pure state $\varrho_U = \kb{\Omega}{\Omega}$ using the selfadjoint operators $\aa$. Since
\begin{align*}
\tr{(\vec{m}_j\cdot \vec{\sigma} \otimes \id)\varrho_U} & = \tr{(U^*\vec{m}_j\cdot \vec{\sigma}U \otimes \id)\kb{\Omega_0}{\Omega_0}}\\
& = \frac{1}{2}\tr{U^* \vec{m}_j\cdot \vec{\sigma}U} =0
\end{align*}
and similarly 
\begin{align*}
\tr{(\id\otimes \vec{n}_j\cdot \vec{\sigma})\varrho_U} & = \tr{(\id\otimes \vec{n}_j\cdot \vec{\sigma})\kb{\Omega_0}{\Omega_0}} = \frac{1}{2}\tr{\vec{n}_j\cdot \vec{\sigma}} =0 \, , 
\end{align*}
we find that the only relevant information that can be extracted is 
\begin{eqnarray*}
\tr{(\vec{m}_j\cdot \vec{\sigma}\otimes \vec{n}_j\cdot \vec{\sigma}) \varrho_U}  &=& \tr{(U^*\vec{m}_j\cdot \vec{\sigma}U\otimes \vec{n}_j\cdot \vec{\sigma}) \kb{\Omega_0}{\Omega_0}} \\
& = & \tr{(R_U\vec{m}_j\cdot \vec{\sigma}\otimes \vec{n}_j\cdot \vec{\sigma}) \kb{\Omega_0}{\Omega_0}} \\
&=& h_j(R_U) 
\end{eqnarray*}
where $h_j$ is defined on the linear space $M_3(\R)$ of real $3\times 3$ matrices and is given by
\begin{align*}
h_j(A) = \tr{(A\vec{m}_j\cdot \vec{\sigma}\otimes \vec{n}_j\cdot \vec{\sigma}) \kb{\Omega_0}{\Omega_0}} \, .
\end{align*}

If the selfadjoint operators $\aa$ could distinguish all pure states, then, in particular, they could distinguish all maximally entangled pure states. Therefore, the linear map $f : M_3 (\R) \to \R^4$ given by
\begin{align*}
f(A) = (h_1(A),h_2(A),h_3(A),h_4(A))^T
\end{align*}
would restrict to an injective map $\tilde{f}:SO(3)\to \R^4$. By Proposition \ref{prop:app3} in Appendix \ref{sec:appendix}, such a map would then be a smooth embedding of $SO(3)$ into $\R^4$. Since $SO(3)$ is diffeomorphic to the real projective $3$-dimensional space $RP^3$ \cite[Proposition 5.2.10]{AbrMarRat88}, and $RP^3$ cannot be embedded into $\R^4$ by \cite{Mah62,Lev63}, we then obtain a contradiction. We thus conclude that the selfadjoint operators $\aa$ cannot distinguish all pure states.

\section{Spin-1 -- It is not only about number of bases}\label{sec:spin1}

As we have now seen, it is enough to measure four bases in order to distinguish all pure states.
However, not all sets of four bases do this as was implied by our consideration of product bases in dimension 4. In this section we demonstrate this further by giving another example where  some natural candidates for the bases fail to distinguish all pure states.
We consider the problem of determining all pure states of a spin-$1$ system by measuring four orthonormal bases corresponding to different spin directions.

The Hilbert space of the spin-$1$ quantum system is $\hh=\C^3$, and the usual spin operators along the three axes are
\begin{align*}
L_x & =
\frac{1}{\sqrt{2}}\left( \begin {array}{ccc} 0&1 &0\\ 1&0&1\\ 0&1&0
\end {array} \right) \, , &
L_y & = \frac{1}{\sqrt{2}}\left( \begin {array}{ccc} 0&-i&0\\ i&0&-i\\ 0&i&0\end {array} \right) \, ,  &
L_z & =\left(\begin{array}{ccc}
1 & 0 & 0\\
0 & 0 & 0\\
0 & 0 & -1
\end{array}\right) \, .
\end{align*}
If $\vec{n}\in\R^3$ is any unit vector, the spin operator along the direction $\vec{n}$ is $\vec{n}\cdot\vec{L} = n_x L_x + n_y L_y + n_z L_z$, which is selfadjoint and has eigenvalues $\{+1,0,-1\}$.
We denote the corresponding eigenprojections as $P^{\vec{n}}_j$, and they can be written as 
\begin{align*}
P^{\vec{n}}_{+1} & = \frac{(\vec{n}\cdot\vec{L})^2 + \vec{n}\cdot\vec{L}}{2} & P^{\vec{n}}_{-1} & = \frac{(\vec{n}\cdot\vec{L})^2 - \vec{n}\cdot\vec{L}}{2} \\
P^{\vec{n}}_{0} & = \id - (\vec{n}\cdot\vec{L})^2 .
\end{align*}
Note that these projections span the same linear space as the operators $\id$, $\vec{n}\cdot\vec{L}$, and $(\vec{n}\cdot\vec{L})^2$.

The next result shows that, if the four directions $\vec{n}_1,\ldots,\vec{n}_4$ are suitably chosen, then the corresponding spin measurements actually determine all  pure states. This refines the lower bound of \cite{SiSt90}, where the authors prove that all pure states of a spin-$1$ system are uniquely determined by six spin components, and should be compared with \cite{Weigert92,FlSiCa05}, where it is shown that, for any spin-$s$ system, the spin observables along two infinitesimally near directions $\vec{n}$ and $\vec{n}'$ together with the expectation value of the spin observable orthogonal to $\vec{n}$ and $\vec{n}'$ are enough to determine all pure states up to a set of measure zero.

\begin{proposition}\label{prop:PIC-spin}
Let $\vec{n}_1,\ldots,\vec{n}_4\in\R^3$ be four directions with the components $\vec{n}_k = (n_{kx},\, n_{ky},\, n_{kz})^T$, and denote by $\bb_k$ the eigenbases of the spin operator along $\vec{n}_k$ for a spin-$1$ quantum system.
Then the bases $\bb_1,\ldots,\bb_4$ can distinguish all pure states if and only if the following conditions hold:
\begin{enumerate}[(a)]
\item the $4\times 5$ real matrix $M$ with entries
\begin{equation}\label{eq:matrix}
\begin{aligned}
M_{k,1} & = 2\sqrt{2} n_{kx} n_{kz} & \qquad
M_{k,2} & = -2\sqrt{2} n_{ky} n_{kz} \\
M_{k,3} & = n_{kx}^2 - n_{ky}^2 & \qquad M_{k,4} & = -2n_{kx}n_{ky} \\
M_{k,5} & = 3n_{kz}^2 - 1
\end{aligned}
\end{equation}
has rank $4$;
\item there exists a nonzero real solution $x = (x_1,\ldots,x_5)^T\in\R^5$ of the linear system $Mx = 0$ such that
\begin{equation}\label{eq:determ}
2 x_1 x_2 x_4 + x_3 ( x_1^2 - x_2^2 ) + x_5 ( x_1^2 + x_2^2 + x_5^2 - x_3^2 - x_4^2 ) \neq 0 .
\end{equation}
\end{enumerate}
\end{proposition}
\begin{proof}
We will show that conditions (a) and (b) in the previous statement are equivalent to the condition in item (ii) of Proposition \ref{prop:critdim3}.\\
The operator system $\h R(\bb_1,\ldots,\bb_4)$ is spanned by $\id$ and the set of $8$ selfadjoint operators
\begin{equation*}
\aa = \{\vec{n}_1\cdot\vec{L},\ldots, \vec{n}_4\cdot\vec{L}, (\vec{n}_1\cdot\vec{L})^2, \ldots, (\vec{n}_4\cdot\vec{L})^2\} \,.
\end{equation*}
Since the operators $\{\vec{n}_1\cdot\vec{L},\ldots, \vec{n}_4\cdot\vec{L}\}$ are linearly dependent, we have $\dim\h R(\bb_1,\ldots,\bb_4)\leq 8$, that is, $\dim\h R(\bb_1,\ldots,\bb_4)^\perp\geq 1$. Hence by Proposition \ref{prop:critdim3}.(ii) we actually need to show that conditions (a) and (b) are equivalent to
\begin{enumerate}[(a')]
\item $\dim \h R(\bb_1,\ldots,\bb_4)^\perp = 1$;
\item there exists an invertible selfadjoint operator $T\in\h R(\bb_1,\ldots,\bb_4)^\perp$.
\end{enumerate}
We begin by showing that conditions (a) and (a') are equivalent, and, when (a) holds,
\begin{equation}\label{eq:RPort}
\h R(\bb_1,\ldots,\bb_4)^\perp = \Phi(\ker(M)) ,
\end{equation}
where $\Phi:\C^5\to M_3(\C)$ is the linear map given by
$$
\Phi\left(\begin{array}{c}
x_1\\ \vdots \\ x_5
\end{array}\right) =
\left(\begin{array}{ccc}
x_5 & x_1 + ix_2 & x_3 + ix_4\\
x_1 - ix_2 & -2x_5 & -x_1 - ix_2\\
x_3 - ix_4 & -x_1 + ix_2 & x_5
\end{array}\right) \,.
$$
Note that $\Phi$ is injective, and $\Phi(\C^5) = \{\id,L_x,L_y,L_z\}^\perp$. Moreover, it is easy to check that $\tr{(\vec{n}_i\cdot\vec{L})^2 \Phi(x)} = m_i x$ for all $x\in\C^5$, where $m_i$ is the $i$th row of the matrix $M$ defined in \eqref{eq:matrix}, and $m_i x$ is the usual matrix product. Hence,
\begin{equation}\label{eq:TkerA}
\Phi(\ker(M)) = \{\id,L_x,L_y,L_z\}^\perp \cap \{(\vec{n}_1\cdot\vec{L})^2,\ldots,(\vec{n}_4\cdot\vec{L})^2\}^\perp = \h R (\aa^\prime)^\perp
\end{equation}
where
\begin{equation*}
\aa' = \{L_x,L_y,L_z,(\vec{n}_1\cdot\vec{L})^2,\ldots,(\vec{n}_4\cdot\vec{L})^2\} \,.
\end{equation*}
Now, suppose that $\rank(M)<4$. Then, $\dim\ker(M)>1$, from which it follows that $\dim\h R (\aa^\prime)^\perp > 1$ by injectivity of $\Phi$. 
Since $\h R(\bb_1,\ldots,\bb_4) = \h R (\aa)$ and $\aa\subset\aa'$, we have $\h R (\aa^\prime)^\perp \subset \h R(\bb_1,\ldots,\bb_4)^\perp$, and this implies that condition (a') does not hold.\\
Conversely, assume that $\rank(M)=4$. We claim that in this case the four unit vectors $\vec{n}_1,\ldots,\vec{n}_4$ span the whole space $\R^3$. 
Indeed, if e.g.~$\vec{n}_i = \alpha_i\vec{n}_1 + \beta_i\vec{n}_2$ for $i=3,4$, then we would have 
$$
(\vec{n}_i\cdot\vec{L})^2 = \alpha_i^2 (\vec{n}_1\cdot\vec{L})^2 + \beta_i^2 (\vec{n}_2\cdot\vec{L})^2 + \alpha_i \beta_i \left\{ \vec{n}_1\cdot\vec{L},\,\vec{n}_2\cdot\vec{L}\right\} \, , 
$$ 
where $\left\{\cdot,\cdot\right\}$ is the anticommutator. Hence, $m_i = \alpha_i^2 m_1 + \beta_i^2 m_2 + \alpha_i \beta_i u$, where $u^T\in\C^5$ (actually, $u^T\in\R^5$) is defined by 
$$
u x = \tr{\left\{\vec{n}_1\cdot\vec{L},\,\vec{n}_2\cdot\vec{L}\right\} \Phi(x)}
$$ 
for all $x\in\C^5$. Thus, $\rank(M)\leq 3$, which is a contradiction. Our claim then implies $\h R(\aa) = \h R(\aa')$. Taking the orthogonal complement, we have $\h R(\aa)^\perp = \h R(\aa')^\perp = \Phi(\ker(M))$ by \eqref{eq:TkerA}, and \eqref{eq:RPort} follows since $\h R(\bb_1,\ldots,\bb_4) = \h R(\aa)$. In particular, by injectivity of $\Phi$, $\dim\h R(\bb_1,\ldots,\bb_4)^\perp = \dim\ker(M) = 1$, that is, condition (a').\\
Finally, assuming (a), we come to the proof of (b) $\Leftrightarrow$ (b'). First of all, observe that the operator $\Phi(x)$ is selfadjoint if and only if $x\in\R^5$. By \eqref{eq:RPort}, condition (b') then amounts to
$$
\det\left(\begin{array}{ccc}
x_5 & x_1 + ix_2 & x_3 + ix_4\\
x_1 - ix_2 & -2x_5 & -x_1 - ix_2\\
x_3 - ix_4 & -x_1 + ix_2 & x_5
\end{array}\right) \neq 0
$$
for some nonzero $x=(x_1,\ldots,x_5)^T\in\R^5$ such that $Mx=0$, that is, condition (b).
\end{proof}

By Proposition \ref{prop:PIC-spin}, it is easy to check that the three orthogonal spin directions $\vec{e}_1$, $\vec{e}_2$, and $\vec{e}_3$ cannot be completed to a set of four directions which would allow unique determination of pure states. Indeed, if $\vec{n}_1 = \vec{e}_1$, $\vec{n}_2 = \vec{e}_2$, $\vec{n}_3 = \vec{e}_3$ and $\vec{n}_4$ is any direction, then
$$
\rank(M) = \rank \left(\begin{array}{ccccc}
0 & 0 & 1 & 0 & -1\\
0 & 0 & -1 & 0 & -1\\
0 & 0 & 0 & 0 & 2\\
\ast & \ast & \ast & \ast & \ast
\end{array}\right) \leq 3
$$
thus contradicting condition (a).

However, there exist sets of four spin directions which distinguish all pure states.  
For example, it is easy to check that the unit vectors
\begin{align*}
\vec{n}_1 & = (0,\,0,\,1)^T & \qquad \vec{n}_2 & = (1/\sqrt{2},\,1/\sqrt{2},\,0)^T \\
\vec{n}_3 & = (1/\sqrt{2},\,0,\,1/\sqrt{2})^T& \qquad \vec{n}_4 & = (0,\,\sqrt{3}/2,\,1/2)^T
\end{align*}
satisfy both conditions (a) and (b) of Proposition \ref{prop:PIC-spin}, hence the corresponding bases can distinguish all pure states.

Finally, we remark that the property of distinguishing all pure states is robust against small perturbations of the unit vectors $\vec{n}_1,\ldots,\vec{n}_4$. Indeed, suppose that conditions (a) and (b) of Proposition \ref{prop:PIC-spin} are satisfied by the four directions $\vec{n}_1^0,\ldots,\vec{n}_4^0$. The matrix $M=M(\vec{n}_1,\ldots,\vec{n}_4)$ defined in \eqref{eq:matrix} is a continuous function of $(\vec{n}_1,\ldots,\vec{n}_4)$, hence so are the diagonalizable matrix-valued functions $M^*M$ and $MM^*$. By condition (a), the $4\times 4$ matrix $MM^*(\vec{n}_1^0,\ldots,\vec{n}_4^0)$ is invertible, hence $MM^*$ is invertible in a neighborhood of $(\vec{n}_1^0,\ldots,\vec{n}_4^0)$, that is, $\rank(M)=4$ in that neighborhood. Thus, condition (a) still holds around $(\vec{n}_1^0,\ldots,\vec{n}_4^0)$. Now, let $x\in\R^5$ be a nonzero real solution of $M(\vec{n}_1^0,\ldots,\vec{n}_4^0)x =0$ which satisfies \eqref{eq:determ}. Moreover, let $Q$ be the orthogonal projection onto $\ker (M^*M)$, and define $\tilde{x}(\vec{n}_1,\ldots,\vec{n}_4) = Q(\vec{n}_1,\ldots,\vec{n}_4)x$. By \cite[Theorem II.5.1]{PTLO66}, $Q$ is a continuous function of $(\vec{n}_1,\ldots,\vec{n}_4)$ in a neighborhood of $(\vec{n}_1^0,\ldots,\vec{n}_4^0)$. As $M^*M$ is a real matrix, also $Q$ is real. Combining these two facts, the map $\tilde{x}$ is a nonzero continuous $\R^5$-valued function such that $\tilde{x}\in\ker (M^*M) = \ker(M)$ and $\tilde{x}(\vec{n}_1^0,\ldots,\vec{n}_4^0) = x$. By continuity, $\tilde{x}$ satisfies \eqref{eq:determ} around $(\vec{n}_1^0,\ldots,\vec{n}_4^0)$. 
Therefore, also condition (b) of Proposition \ref{prop:PIC-spin} remains fulfilled for small perturbations of $(\vec{n}_1^0,\ldots,\vec{n}_4^0)$. We conclude that, if the eigenbases of the spin operators along the directions $\vec{n}_1^0,\ldots,\vec{n}_4^0$ distinguish all pure states, then this still holds true in a neighborood of these directions.

\section{Remarks on the question in infinite dimensional Hilbert space}\label{sec:discussion}

The question of the title is  meaningful also in the case of an infinite dimensional Hilbert space. The proof of Proposition \ref{prop:2bases} works without changes also in that case, so we conclude that two orthonormal bases cannot distinguish all pure states even when $d=\infty$.
However, the construction of four orthonormal bases in Sec. \ref{sec:4bases} has no direct generalization to the infinite dimensional Hilbert space.
We are, in fact, not aware of a construction in infinite dimension that would give a finite number of bases that can distinguish all pure states.

In the infinite dimensional case it is natural to allow also measurements of continuous observables such as position $Q$ and momentum $P$. In fact, the determination of pure states from the statistics of such measurements was precisely what was addressed in the original Pauli problem. Since it is known that $Q$ and $P$ alone are not sufficient for this purpose, we can ask if this set can be suitably completed to make it able to distinguish all pure states. 
One natural attempt to obtain such a completion would be to add rotated quadratures $Q_\theta = \cos\theta\, Q + \sin\theta\, P$. 
It is known that by measuring all of the quadratures, it is possible to determine an arbitrary state, pure or mixed \cite{VoRi89}, but it is not known if a smaller subset is sufficient for pure state determination. It was recently shown that no triple $\{Q_{\theta_1}, Q_{\theta_2}, Q_{\theta_3} \}$ of quadratures is enough \cite{CaHeScTo14}. In particular, it is not sufficient to measure position, momentum, and a single additional rotated quadrature. 

For larger, but still finite, sets of rotated quadratures it is only known that if such a set can distinguish all pure states, then the choice of the angles $\theta$ is a delicate issue. In \cite{CaHeScTo14}, it was shown that if a finite set of angles $\theta_1,\ldots, \theta_n$ satisfies $\theta_i-\theta_j\in\rational \, \pi$ for all $i,j=1,\ldots, n$, then the corresponding observables are not sufficient for pure state determination. In \cite{AnJa15}, a similar result was proved in the case that $\cot\theta_j\in\rational$ for all $j=1,\ldots,n$.

\section*{Acknowledgements.}

JS and AT acknowledge financial support from the Italian Ministry of Education, University and Research (FIRB project RBFR10COAQ).

\section*{Appendix A. Linear smooth embeddings in $\R^n$}\label{sec:appendix}

If $M$ is a real differentiable manifold and $x\in M$, we denote by $T_x(M)$ the (real) tangent space of $M$ at $x$. If $N$ is another differentiable manifold and $f:M\to N$ is a differentiable map, we let $\de f_x : T_x(M) \to T_{f(x)}(N)$ be the differential of $f$ at $x$.  

{In this section we assume that  $M$ is a submanifold of a real linear space $V$, and  
we prove that, if a linear map $f:V\to\R^n$ restricts to an injective map $\tilde{f}:M\to\R^n$, then $\tilde{f}$ is a smooth embedding in the following two cases:}
\begin{enumerate}
\item $V=S_d(\C)$ is the linear space of complex $d\times d$ selfadjoint matrices and $M=\cP$ is the submanifold of pure states (see \cite[Section 4.2]{HeMaWo13});
\item $V=M_3(\R)$ is the linear space of real $3\times 3$ matrices and $M=SO(3)$ is the submanifold of orthogonal matrices with unit determinant.
\end{enumerate}

The next two results are Theorem 5 and a particular case of Lemma 1 in \cite{HeMaWo13}. Up to our knowledge, Proposition~\ref{prop:app3} below is new.

\begin{lemma}\label{prop:app1}
Let $V$, $W$ be two real linear spaces, and $M$ a compact submanifold of $V$. Suppose that $T_x(M) \subseteq \R(M-M)$ for all $x\in M$. Then, if a linear map $f:V\to W$ restricts to an injective map $\tilde{f} : M \to W$, the restriction $\tilde{f}$ is a smooth embedding of $M$ in $V$.
\end{lemma}
\begin{proof}
Since $M$ is compact and $\tilde{f}$ is continuous, injectivity implies that $\tilde{f}$ is a homeomorphism of $M$ onto $\tilde{f}(M)$ \cite[Proposition 1.6.8]{AN89}. In order to show that it is a smooth embedding, it remains to prove that $\de\tilde{f}_x$ is injective for all $x\in M$.\\
Note that by linearity $\de\tilde{f}_x = \left.f\right|_{T_x(M)}$ for all $x\in M$. Thus, if $u\in T_x(M)$ with $u=\lam(x_1-x_2)$ for some $\lam\in\R$ and $x_1,x_2\in M$, then
$$
\de\tilde{f}_x u = f(u) = \lam(f(x_1) - f(x_2)) = \lam(\tilde{f}(x_1) - \tilde{f}(x_2)) \equiv 0
$$
if and only if $\lam = 0$ or $\tilde{f}(x_1) = \tilde{f}(x_2)$, that is, $x_1 = x_2$ by injectivity of $\tilde{f}$. In both cases, we have $u=0$, hence $\de\tilde{f}_x$ is injective as claimed.
\end{proof}

\begin{proposition}\label{prop:app2}
If $\tilde{f}:\cP\to\R^n$ is injective and it is the restriction of a linear map $f:S_d(\C)\to\R^n$, then $\tilde{f}$ is a smooth embedding of $\cP$ in $\R^n$.
\end{proposition}
\begin{proof}
By Lemma \ref{prop:app1}, it suffices to prove that $T_\varrho(\cP) \subseteq \R(\cP-\cP)$ for all pure states $\varrho\in\cP$. Indeed, $\cP$ is an orbit for the adjoint action of the group $SU(d)$ on the linear space $S_d(\C)$, hence
$$
T_\varrho(\cP) = \left\{\frac{\de}{\de t} \e^{itH}\varrho\e^{-itH} \mid H\in S_d(\C)\right\} = \left\{i[H,\varrho] \mid H\in S_d(\C)\right\} \,.
$$
If $\varrho$ is a pure state and $H\in S_d(\C)$, then $i[H,\varrho]$ is a selfadjoint traceless matrix with rank at most $2$. Therefore, $i[H,\varrho] = \lam(\kb{\psi^+}{\psi^+} - \kb{\psi^-}{\psi^-})$ for some $\lam\in\R$ and unit vectors $\psi^+, \psi^-\in\C^d$, which proves the claim.
\end{proof}

\begin{proposition}\label{prop:app3}
If $\tilde{f}:SO(3)\to\R^n$ is injective and it is the restriction of a linear map $f:M_3(\R)\to\R^n$, then $\tilde{f}$ is a smooth embedding of $SO(3)$ in $\R^n$.
\end{proposition}
\begin{proof}
Again, by Lemma \ref{prop:app1} it is enough to prove that $T_R(SO(3)) = \R(SO(3)-SO(3))$ for all $R\in SO(3)$.\\
Denote by $M_3^-(\R)$ the linear subspace of antisymmetric matrices in $M_3(\R)$. Then, $T_R(SO(3)) = RM_3^-(\R)$ for all $R\in SO(3)$. We claim that any $X\in M_3^-(\R)$ can be written $X = \lam(R_0-R_0^T)$ for some $R_0\in SO(3)$ and $\lam\in\R$. Indeed, the map $g : SO(3)\to M_3^-(\R)$ with $g(R) = R-R^T$ is a diffeomorphism of an open neighborhood $U$ of the identity $I$ onto a neighborhood  $g(U)$ of $0$, since its differential
$$
\de g_I (X) = \left.\frac{\de}{\de t} (\exp(tX) - \exp(-tX))\right|_{t=0} = 2X
$$
is bijective. It follows that $\R g(U) = M_3^-(\R)$, hence the claim.
\end{proof}


\begin{thebibliography}{10}

\bibitem{Pauli33}
W.~Pauli.
\newblock {\em Die allgemeinen {P}rinzipen der {W}ellenmechanik}.
\newblock in: H.Geiger and K. Scheel (Eds.), \emph{{H}andbuch der {P}hysik},
  Vol. 24. Springer-Verlag, Berlin, 1933.

\bibitem{PFQM44}
H.~Reichenbach.
\newblock {\em Philosophic Foundations of Quantum Mechanics}.
\newblock University of California Press, Berkeley, 1944.

\bibitem{Vogt78}
A.~Vogt.
\newblock {\em Position and momentum distributions do not determine the quantum
  mechanical state}.
\newblock in: {A}. {R}. {M}arlow (ed.), \emph{{M}athematical {F}oundations of
  {Q}uantum {T}heory}. Academic Press, New York, 1978.

\bibitem{Moroz83}
B.Z. Moroz.
\newblock Reflections on quantum logic.
\newblock {\em Internat. J. Theoret. Phys.}, 22:329--340, 1983.

\bibitem{Moroz84}
B.Z. Moroz.
\newblock Erratum: ``{R}eflections on quantum logic'' [{I}nternat. {J}.
  {T}heoret. {P}hys. {\bf 22}\ (1983), no. 4, 329--340; {MR}0701315
  (84i:81009)].
\newblock {\em Internat. J. Theoret. Phys.}, 23:497--498, 1984.

\bibitem{MoPe94}
B.Z. Moroz and A.M. Perelomov.
\newblock On a problem posed by {P}auli.
\newblock {\em Theoretical and Mathematical Physics}, 101:1200--1204, 1994.

\bibitem{HeMaWo13}
T.~Heinosaari, L.~Mazzarella, and M.M. Wolf.
\newblock Quantum tomography under prior information.
\newblock {\em Comm. Math. Phys.}, 318:355--374, 2013.

\bibitem{MoVo13}
D.~Mondragon and V.~Voroninski.
\newblock Determination of all pure quantum states from a minimal number of
  observables.
\newblock {\em arXiv:1306.1214 [math-ph]}, 2013.

\bibitem{Jaming14}
P.~Jaming.
\newblock Uniqueness results in an extension of {P}auliʼs phase retrieval
  problem.
\newblock {\em Applied and Computational Harmonic Analysis}, 37:413--441, 2014.

\bibitem{BuLa89}
P.~Busch and P.~Lahti.
\newblock The determination of the past and the future of a physical system in
  quantum mechanics.
\newblock {\em Found. Phys.}, 19:633--678, 1989.

\bibitem{KoPa13}
Andrzej Komisarski and Adam Paszkiewicz.
\newblock On a system of measurements which is complete in a statistical sense.
\newblock {\em Infinite Dimensional Analysis, Quantum Probability and Related
  Topics}, 16(03):1350026, 2013.

\bibitem{Finkelstein04}
J.~Finkelstein.
\newblock Pure-state informationally complete and ``really'' complete
  measurements.
\newblock {\em Phys. Rev. A}, 70:052107, 2004.

\bibitem{Szego}
G.~Szeg\"o.
\newblock {\em Orthogonal Polynomials}.
\newblock American Mathematical Society, 4th edition, 1975.

\bibitem{KeVrWo15}
M.~Kech, P.~Vrana, and M.~M. Wolf.
\newblock The role of topology in quantum tomography.
\newblock {\em arXiv:1503.00506 [quant-ph]}, 2015.

\bibitem{Lee2009}
J.M. Lee.
\newblock {\em Manifolds and {D}ifferential {G}eometry}, volume 107 of {\em
  Graduate Studies in Mathematics}.
\newblock American Mathematical Society, Providence, RI, 2009.

\bibitem{Mayer65}
K.H. Mayer.
\newblock Elliptische {D}ifferentialoperatoren und {G}anzzahligkeitss\"atze
  f\"ur charakteristische {Z}ahlen.
\newblock {\em Topology}, 4:295--313, 1965.

\bibitem{Vollbrecht2000}
K.~G.~H. Vollbrecht and R.~F. Werner.
\newblock Why two qubits are special.
\newblock {\em J. Math. Phys.}, 41(10):6772--6782, 2000.

\bibitem{AbrMarRat88}
R.~Abraham, J.~E. Marsden, and T.~Ratiu.
\newblock {\em Manifolds, tensor analysis, and applications}, volume~75 of {\em
  Applied Mathematical Sciences}.
\newblock Springer-Verlag, second edition, 1988.

\bibitem{Mah62}
M.~Mahowald.
\newblock On the embeddability of the real projective spaces.
\newblock {\em Proc. Amer. Math. Soc.}, 13:763--764, 1962.

\bibitem{Lev63}
J.~Levine.
\newblock Imbedding and immersion of real projective spaces.
\newblock {\em Proc. Amer. Math. Soc.}, 14:801--803, 1963.

\bibitem{SiSt90}
W.~Stulpe and M.~Singer.
\newblock Some remarks on the determination of quantum states by measurements.
\newblock {\em Found. Phys. Lett.}, 3:153--166, 1990.

\bibitem{Weigert92}
S.~Weigert.
\newblock Pauli problem for a spin of arbitrary length: {A} simple method to
  determine its wave function.
\newblock {\em Phys. Rev. A}, 45:7688--7696, 1992.

\bibitem{FlSiCa05}
S.~Flammia, A.~Silberfarb, and C.~Caves.
\newblock Minimal informationally complete measurements for pure states.
\newblock {\em Found. Phys.}, 35:1985--2006, 2005.

\bibitem{PTLO66}
T.~Kato.
\newblock {\em Perturbation theory for linear operators}.
\newblock Springer-Verlag, Berlin, 1995.
\newblock Reprint of the 1980 edition.

\bibitem{VoRi89}
K.~Vogel and H.~Risken.
\newblock Determination of quasiprobability distributions in terms of
  probability distributions for the rotated quadrature phase.
\newblock {\em Phys. Rev. A}, 40:2847--2849, 1989.

\bibitem{CaHeScTo14}
C.~Carmeli, T.~Heinosaari, J.~Schultz, and A.~Toigo.
\newblock Non-uniqueness of phase retrieval for three fractional fourier
  transforma.
\newblock {\em Appl. Comput. Harm. Anal.}, 2014.
\newblock In Press. Available online.

\bibitem{AnJa15}
S.~Anreys and P.~Jaming.
\newblock Zak transform and non-uniqueness in an extension of {P}auli's phase
  retrieval problem.
\newblock {\em arXiv:1501.03905 [math.CA]}, 2015.

\bibitem{AN89}
G.~Pedersen.
\newblock {\em Analysis now}.
\newblock Springer-Verlag, New York, 1989.

\end{thebibliography}
\end{document}